\documentclass[aps,prl,reprint,superscriptaddress,amsmath,amssymb,longbibliography]{revtex4-2}
\usepackage[T1]{fontenc}
\usepackage{lmodern}
\usepackage{bm}
\usepackage{graphicx}
\usepackage{amsthm}
\usepackage{hyperref}
\hypersetup{colorlinks=true,citecolor=blue,linkcolor=blue,urlcolor=blue}
\newcommand{\Tr}{\operatorname{Tr}}
\newcommand{\Ns}{N^{*}}
\newcommand{\F}{\mathbb F_2}
\newcommand{\Acc}{\mathsf{Acc}}
\newcommand{\eps}{\varepsilon_0}
\newcommand{\ket}[1]{\lvert #1\rangle}
\newcommand{\bra}[1]{\langle #1\rvert}
\newcommand{\ip}[2]{\langle #1,#2\rangle_2}
\newcommand{\prlheading}[1]{\par\indent\textit{#1}---\ignorespaces}
\newtheoremstyle{prlstatement}{3pt}{3pt}{\normalfont}{\parindent}
  {\normalfont\itshape}{---}{0pt}{}
\theoremstyle{prlstatement}
\newtheorem{theorem}{Theorem}
\theoremstyle{plain}
\newtheorem{lemma}{Lemma}
\newtheorem{proposition}{Proposition}
\begin{document}
\title{Finite-Size Limits of Adaptive Pauli Prediction}
\author{Luofei Wang}
\thanks{These authors contributed equally.}
\affiliation{School of Physics, Beijing Institute of Technology, Beijing, China}
\author{Boyu Liu}
\thanks{These authors contributed equally.}
\affiliation{Gaoling School of Artificial Intelligence, Renmin University of China, Beijing, China}
\author{Zihe Wang}
\email{wang.zihe@ruc.edu.cn}
\affiliation{Gaoling School of Artificial Intelligence, Renmin University of China, Beijing, China}
\author{Zhangqi Yin}
\email{zqyin@bit.edu.cn}
\affiliation{School of Physics, Beijing Institute of Technology, Beijing, China}
\begin{abstract}
The expectation values of a predetermined set of Pauli observables can be predicted using only a few copies of a quantum state. However, when later observables depend on earlier predictions, how the system size constrains the required number of copies remains unclear. We use classical interactive fingerprinting codes to establish this constraint. A compact Pauli-parity encoding embeds the fingerprinting interaction into the quantum prediction task, allowing any quantum prediction mechanism on the constructed diagonal states to be simulated by a classical response strategy subject to the same information constraints. We show that accurate prediction from too few copies would violate the guarantees of the fingerprinting code. For $n$ qubits and $M$ adaptively chosen Pauli observables, requiring all predictions to be simultaneously accurate to a fixed additive error with probability at least $2/3$ yields the worst-case lower bound $N^*(n,M)=\Omega\!\left(\sqrt{\min\{n,M\}}\right)$ on the required number of copies. For $M\le n$, this lower bound matches known upper bounds. In particular, $n$ adaptive predictions require $\Theta(\sqrt{n})$ copies, whereas the same number of predetermined queries requires only $O(\log n)$ copies. This separation persists even for diagonal states and commuting observables.
\end{abstract}
\maketitle

\prlheading{Introduction}
Estimating selected expectation values can require far fewer copies than
reconstructing an unknown quantum state. Full tomography determines
exponentially many parameters in a general many-qubit system, whereas
prediction can focus on the observables relevant to an experiment
\cite{Haah17,ODonnell16,Aaronson18,HKP20}.
For a fixed list of $M$ Pauli observables, a protocol combining two-copy
Bell measurements with gentle collective sign estimation uses
$O(\log(M+1))$ copies at constant accuracy and confidence
\cite{HKP21}. Existing protocols also address general observable lists
specified in advance; these guarantees make it possible to extract many
properties from a limited supply of identically prepared systems
\cite{JHL26}.

If the next observable is selected using earlier estimates, the query
list is no longer fixed in advance. Rydberg-array experiments extract
two-point occupation correlations and nonlocal $Z$-string parities
from images of prepared many-body states, where $Z_i=1-2n_i$ and
$n_i$ denotes the Rydberg occupation \cite{Semeghini21}.
For one fixed preparation, an adaptive analysis could use estimates of
$\langle Z_iZ_j\rangle$ to select a region $A$ in which to examine
$\langle\prod_{r\in A}Z_r\rangle$.
Photon-polarization entanglement-detection experiments have implemented
decision trees whose next Pauli correlation is chosen from earlier
correlation estimates \cite{Laskowski12}.
In these problems, the preparation remains fixed, while the questions
depend on previously obtained data. Classical adaptive data analysis
studies the resulting generalization problem \cite{Dwork15,Hardt14,SU15,Bassily21,Dagan22}.
Quantum prediction combines this statistical issue with the disturbance
caused by measurement \cite{AaronsonRothblum19,Badescu21}.
A guarantee on the required quantum resources must therefore hold
simultaneously for all queries in the adaptive sequence.

Huang, Lewis, Huang, and Preskill (HLHP) proved that $M$ adaptive Pauli queries on $n$ qubits can require $\Omega(\min\{\sqrt M,n^{1/3}\})$ copies, and gave a protocol using $O(\min\{\sqrt M,n\})$ copies \cite{HLHP26}. For $M=n$, these bounds leave a gap between $n^{1/3}$ and $\sqrt n$. The recent study of online shadow tomography by Chen, O'Donnell, Pelecanos, and Wright gives the further upper bound $O(\sqrt n\log(M+n+1))$ at fixed accuracy \cite{COPW26}. The unresolved size dependence concerns how much data a finite system can force an adaptive procedure to use.

We improve the lower bound to $\Omega(\sqrt{\min\{n,M\}})$. Our main approach reduces Pauli queries to the information-theoretic fingerprinting framework of Steinke--Ullman \cite{SU15}, using a representation in which every query is a single Pauli parity. For $d$ users and $L$ rounds, a computational-basis label contains $d+L$ bits. Figure~\ref{fig:mechanism} shows how the user and round registers select an encrypted bit and its key. We prove that the mechanism's complete adaptive transcript can be simulated with the same distribution using only the information available to a legal coalition. Together with the established HLHP upper bound, this gives a copy complexity of $\Theta(\sqrt n)$ for $n$ adaptive predictions, compared with only $O(\log n)$ copies for a fixed query list.

\begin{figure*}[t]
\centering
\includegraphics[width=\textwidth]{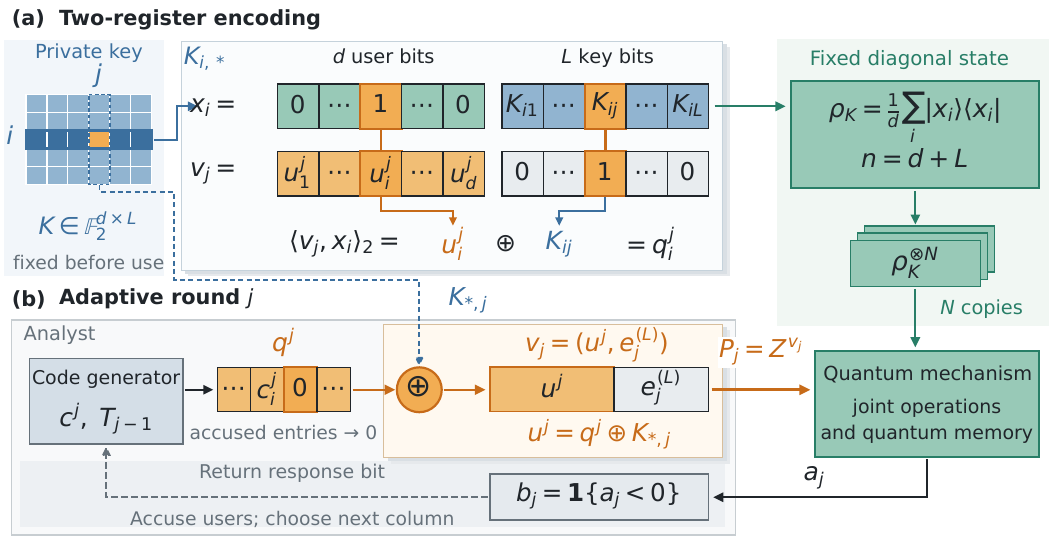}
\caption{Fixed-state realization of an adaptive fingerprinting attack
(schematic). (a) The analyst's private key matrix $K$, fixed before
interaction, defines $x_i=(e_i^{(d)},K_{i,*})$ and the diagonal
state $\rho_K$, on $d+L$ qubits before padding.
The highlighted coordinates give
$\langle v_j,x_i\rangle_2=u_i^j\oplus K_{ij}=q_i^j$.
(b) The analyst sets the code bits at accused positions to zero, masks
the resulting column, and submits one Pauli query. The mechanism
receives the initial state copies and arriving queries. Its answer
determines $b_j$, which updates the accusation set and the next code
column. Dashed arrows represent the private round pad and response
feedback. Register lengths are schematic.}
\label{fig:mechanism}
\end{figure*}

\prlheading{Prediction task}
Let $\mathcal P_n=\{I,X,Y,Z\}^{\otimes n}$. Before interaction begins, a quantum prediction mechanism receives $N$ copies of an unknown state $\rho$. Here and below, we use mechanism as shorthand for a quantum prediction algorithm to distinguish it from a classical prediction algorithm. In round $j$, an analyst submits $P_j\in\mathcal P_n$ and receives $a_j\in[-1,1]$ before selecting the next query. The analyst may know $\rho$ and the mechanism's description, and may use the complete previous transcript. The mechanism accesses $\rho$ through its initial copies and incoming queries. Arbitrary joint operations, quantum memory, and computation are allowed, and the two parties' private randomness is initially independent. For fixed $\eps=10^{-3}$, require
\begin{equation}
 \Pr\!\left[\bigcap_{j=1}^M
 \{|a_j-\Tr(P_j\rho)|\le\eps\}\right]\ge\frac23.
 \label{eq:accuracy}
\end{equation}
Write $\Ns(n,M)$ for the minimum budget of a mechanism satisfying
this condition for every fixed state and legal analyst.
Repeated queries are allowed, and logarithms have base two.

\begin{theorem}\label{thm:fixed}
There is a universal constant $c>0$ such that, for all integers
$n\ge2$ and $M\ge1$,
\begin{equation}
 \Ns(n,M)\ge c\sqrt{\min\{n,M\}}.
 \label{eq:new}
\end{equation}
The bound holds already for computational-basis diagonal states and
queries in $\{I,Z\}^{\otimes n}$.
\end{theorem}

At the same fixed accuracy, protocols from existing studies give $\Ns(n,M)\le C_0\min\{\sqrt M,n\}$: the $\sqrt M$ branch follows from HLHP and the $n$ branch from the HKP full-Pauli table \cite{HLHP26,HKP21}. Consequently, $\Ns(n,M)=\Theta(\sqrt M)$ holds throughout $1\le M\le n$. The construction below explains this lower bound. The Supplemental Material gives complete proofs, uniform constants, and the conversions of the cited upper bounds \cite{Supplemental}.

\prlheading{A Pauli representation of the attack}
An interactive fingerprinting code generates a column of bits for $d$
users in each round. A coalition sees its active members' bits and
returns one response bit; the code uses these responses to accuse users,
whose future bits are withheld.
A response is inconsistent if it agrees with no entry of the full
column.
For a coalition of $s$ users, the Steinke--Ullman code has $d=2000s$
and $L\le C_Fs^2$. For $s\ge10$, the probability that all responses
remain consistent or that more than $(d-s)/2000$ nonmembers are
accused is at most $\eta_s<1/3$ \cite{SU15}.

Draw the key matrix $K\in\F^{d\times L}$ uniformly and prepare
\begin{equation}
 x_i=(e_i^{(d)},K_{i,*}),\qquad
 \rho_K=\frac1d\sum_{i=1}^d\ket{x_i}\bra{x_i}.
 \label{eq:state}
\end{equation}
Here $e_i^{(d)}$ is a coordinate vector.
The analyst knows $K$, while the mechanism receives copies of $\rho_K$.
Let $T_{j-1}$ be the users accused before round $j$.
Set $q_i^j=c_i^j$ for unaccused users and $q_i^j=0$ otherwise, where
$c^j$ is the code column. The query is
\begin{equation}
 u^j=q^j\oplus K_{*,j},\qquad
 v_j=(u^j,e_j^{(L)}),\qquad P_j=Z^{v_j}.
 \label{eq:query}
\end{equation}
The first register selects $u_i^j$, and the round register selects
$K_{ij}$. Their binary sum cancels the mask. With
\mbox{$Z^v=\bigotimes_r Z^{v_r}$}, this gives
\begin{equation}
 \begin{split}
 \ip{v_j}{x_i}&=u_i^j\oplus K_{ij}=q_i^j,\\
 \mu_j:=\Tr(P_j\rho_K)&=\frac1d\sum_i(-1)^{q_i^j}.
 \end{split}
 \label{eq:decode}
\end{equation}
The mechanism's answer $a_j$ supplies the fingerprinting-code response
$b_j=\mathbf1\{a_j<0\}$.

The label length is $d+L$, compared with $\lceil\log_2d\rceil+2dL$ in HLHP's Pauli construction. Steinke--Ullman already obtained $O(s^2)$-bit hard instances for unrestricted statistical queries using records $(i,sk_i)$ \cite{SU15}; independent round masks instantiate this key framework, as~in~Ref.~\cite{NSSSU18}. Here the user-index register uses one-hot encoding, making decryption a single Pauli parity. For exact representations valid for all users, key rows, ciphertext columns, and rounds, a rank count of the coordinate functions gives the minimum length $d+L$. Allowing a known query sign reduces this length to $d+L-1$. This representation-specific statement and its proof are given as a proposition in the Supplemental Material.

\prlheading{Derivation of the lower bound}
To prove Theorem~\ref{thm:fixed}, we must turn any quantum prediction
algorithm into a legal fingerprinting coalition. First draw an ordered
tuple $J=(J_1,\ldots,J_s)$ of independent uniform users. The diagonal
input has the exact decomposition
\begin{equation}
 \rho_K^{\otimes s}
 =\mathbb E_J\!\left[\bigotimes_{k=1}^s
 \ket{x_{J_k}}\bra{x_{J_k}}\right].
 \label{eq:diagonalinput}
\end{equation}
For every measurable transcript event, the corresponding quantum
instrument is linear in the input state.
A classical simulator can therefore reproduce the mechanism's causal
output probabilities from the sampled labels, even for joint
measurements and quantum memory. When the past answers agree, the
analyst selects the same next query, so this equality extends through
the adaptive interaction.

The simulator must also respect the coalition's information.
Let $S=\{J_1,\ldots,J_s\}$ and pad it deterministically to a coalition
$S^+$ of size $s$,
independently of the code generator. Sample the key rows belonging to
$S$. At an active sampled position, the coalition forms
$u_i^j=c_i^j\oplus K_{ij}$ from its received code bit; at an accused
sampled position it uses $u_i^j=K_{ij}$. For the remaining positions,
write $\bar S=[d]\setminus S$. Condition on $J$, sampled key rows,
the generator and simulator random seeds, and the unsampled key
columns used before round $j$; call this information
$\mathcal H_{j-1}$. Causality makes $q^j$ determined by this information,
while $K_{\bar S,j}$ remains independent and uniform. Consequently,
for every $y\in\F^{d-|S|}$,
\begin{equation}
 \Pr[u^j_{\bar S}=y\mid\mathcal H_{j-1}]
 =2^{-(d-|S|)}.
 \label{eq:freshpad}
\end{equation}
The coalition can thus generate the unsampled ciphertext with fresh
fair bits. Induction over rounds gives the same joint distribution of
sample indices, padded coalition, queries, answers, code columns, and
accusations as the real attack.
The full columns in this comparison are proof variables; the coalition
receives only its active members' bits.

Let $\theta$ count inconsistent responses and let $\psi=|T_L\setminus S^+|$ count falsely accused nonmembers. If $\psi\le(d-s)/2000$, the accused users consist of these nonmembers and at most $s$ coalition members, so $|T_{j-1}|/d\le[s+(d-s)/2000]/d=3999/4\,000\,000$. An all-zero code column has $\mu_j=1$. For an all-one code column, the accused positions contribute $+1$ rather than $-1$, giving $\mu_j=-1+2|T_{j-1}|/d\le-0.9980005$. An answer accurate to tolerance $\tau\le1/2$ has the required sign in either case, while a mixed code column accepts either response. Writing $\Acc_L^{(\tau)}(\rho_K)$ for the event of simultaneous accuracy in these rounds, we obtain, for $0\le\tau\le1/2$,
\begin{equation}
 \Acc_L^{(\tau)}(\rho_K)
 \subseteq\{\theta=0\}\cup\{\psi>(d-s)/2000\}.
 \label{eq:containment}
\end{equation}
The event on the right depends on the joint record just simulated,
so the fingerprinting guarantee implies
\begin{equation}
 \mathbb E_K\Pr[\Acc_L^{(\tau)}(\rho_K)\mid K]
 \le\eta_s<1/3.
 \label{eq:keyaverage}
\end{equation}
For every candidate mechanism, some fixed key therefore gives success
below the required $2/3$. That key fixes the state and analyst before
the run.

To fit the construction into the requested system, set $B=2000+C_F$.
Then $d+L\le2000s+C_Fs^2\le Bs^2$ and $L\le Bs^2$.
For $w=\min\{n,M\}\ge4Bs_0^2$, $s_0=10$, choose
\begin{equation}
 s=\left\lfloor\sqrt{\frac{w}{B}}\right\rfloor
 \ge\frac12\sqrt{\frac{w}{B}}\ge s_0.
 \label{eq:fit}
\end{equation}
Since $Bs^2\le w$, both budgets suffice. Pad unused qubits by $\ket0$
and unused queries by identities. A guarantee for all rounds would
imply accuracy on the attack prefix, contradicting
Eq.~\eqref{eq:keyaverage}. Thus $\Ns(n,M)>s$; smaller budgets are also
excluded because extra copies can be discarded.
For bounded $w$, the same first query $Z_1$ has means $+1$ and $-1$
on two fixed product states. A zero-copy mechanism has the same answer
distribution on both, so their disjoint accuracy intervals cannot each
have probability $2/3$. Hence $\Ns(n,M)\ge1$. Choosing a suitable universal constant $c$ now completes the proof of Eq.~\eqref{eq:new}.

\prlheading{Accuracy and query budgets}
We next examine how the copy budget grows as the tolerance decreases. Let $\Ns(n,M,\varepsilon)$ denote the same task at tolerance $\varepsilon$, so $\Ns(n,M)=\Ns(n,M,\eps)$. Write $r=\sqrt{\min\{n,M\}}$ and $\ell=\min\{n,\log(M+1)\}$.

\begin{theorem}\label{thm:precision}
There is a universal constant $c_1>0$ such that, for all
$n\ge2$, $M\ge1$, and $0<\varepsilon\le1/4$,
\begin{equation}
 \Ns(n,M,\varepsilon)\ge
 c_1\max\{r/\varepsilon,\ell/\varepsilon^2\}.
 \label{eq:precisionlower}
\end{equation}
Both terms hold for diagonal states and commuting Pauli queries.
\end{theorem}

The adaptive term follows by mixing the hard state with a known state:
\begin{equation}
 \sigma_{K,p}=(1-p)\ket{0^n}\bra{0^n}+p\rho_K,
 \qquad p=2\varepsilon.
 \label{eq:dilution}
\end{equation}
Every attack parity equals $+1$ on $\ket{0^n}$, so its mean becomes
$1-p+p\mu_j$. Given a diluted answer $A_j$, the analyst returns to the
original scale by
\begin{equation}
 a_j=\Pi_{[-1,1]}\!\left(\frac{A_j-(1-p)}p\right),
 \qquad \varepsilon/p=1/2,
 \label{eq:rescale}
\end{equation}
where $\Pi$ clips to the interval. Clipping cannot increase the error
relative to $\mu_j\in[-1,1]$. Thus accurate diluted answers meet
exactly the tolerance used in Eq.~\eqref{eq:keyaverage}.

In the product mixture $\sigma_{K,p}^{\otimes N}$, the number of hard-state factors satisfies $N_{\mathrm h}\sim\operatorname{Bin}(N,p)$. A simulator with $s$ copies of the hard state measures their basis labels and uses them in the hard-state slots, preparing the remaining slots in the known state. It runs the candidate mechanism's quantum instruments and rescales the answers according to Eq.~\eqref{eq:rescale}. If $N_{\mathrm h}>s$, it aborts the internal simulation and answers zero thereafter. Couple the mixture flags and random seeds. Induction over rounds then shows that, whenever $N_{\mathrm h}\le s$, the queries and rescaled answers agree with those used by the analyst in the diluted problem. Writing $\Acc_M^{(\varepsilon)}(\sigma_{K,p})$ for the event that all answers on the diluted state are simultaneously accurate, Eq.~\eqref{eq:keyaverage} and Markov's inequality give
\begin{equation}
 \begin{split}
 &\mathbb E_K\Pr[\Acc_M^{(\varepsilon)}(\sigma_{K,p})\mid K]\\
 &\quad\le\eta_s+\Pr[N_{\mathrm h}>s]
 \le\eta_s+\frac{Np}{s+1}.
 \end{split}
 \label{eq:successcoupling}
\end{equation}
If $Np\le s/6$, this probability is below $1/2$, so some fixed key
violates the universal guarantee. Every valid budget must therefore
satisfy $N>s/(12\varepsilon)$. Substitution of Eq.~\eqref{eq:fit}
gives the $r/\varepsilon$ term whenever that dimension choice applies.

The identification term uses $R=\min\{M+1,2^n\}$ states:
$\omega_0=I/2^n$ and $\omega_i=(I+4\varepsilon Z^{v_i})/2^n$ for
$1\le i\le R-1$, with distinct nonzero parity vectors. These are valid
quantum states because $4\varepsilon\le1$. Pauli orthogonality gives, for
$1\le i,j\le R-1$,
\begin{equation}
 \Tr(Z^{v_j}\omega_i)=4\varepsilon\delta_{ij},\qquad
 \Tr(Z^{v_j}\omega_0)=0.
 \label{eq:identmeans}
\end{equation}
Use the same fixed query list for every state. On the simultaneous-accuracy
event, the biased entry, if present, is the unique answer exceeding
$2\varepsilon$. If no answer exceeds this threshold, the decoder returns
label zero. Accurate predictions therefore identify the state.
Let $V$ be a uniform state label, $Y$ the answer record, and $p_i$ the
basis distribution of $\omega_i$. The likelihood ratio is
$p_i(x)/p_0(x)=1+4\varepsilon(-1)^{\ip{v_i}{x}}$ for $i\ge1$.
Using $\ln z\le z-1$ gives
$D_2(p_i\Vert p_0)\le16\varepsilon^2/\ln2$, where $D_2$ is relative
entropy in bits. The diagonal simulation, data processing, and
additivity bound the information in $Y$ by $N$ times this quantity.
Decoding with error at most $1/3$ then gives, by Fano's inequality,
\begin{equation}
 \tfrac23\log R-1\le I(V;Y)
 \le\frac{16N\varepsilon^2}{\ln2}.
 \label{eq:fanochain}
\end{equation}
For $R\ge8$, the left side is at least $\tfrac13\log R$.
For $2\le R<8$, the same threshold distinguishes labels zero and one;
data processing to this binary event gives
$N\ge\ln2/(48\varepsilon^2)$. Combining the ranges yields
$N\ge\log R/(216\varepsilon^2)=\ell/(216\varepsilon^2)$.
Since $\ell\ge1$ and $\varepsilon\le1/4$, this also supplies the
adaptive term for the remaining bounded $w$.
Every universal mechanism must obey both lower bounds, proving
Theorem~\ref{thm:precision}. The Supplemental Material gives the
complete coupling, binary test, and uniform constants.

\begin{figure}[t]
\centering
\includegraphics[width=\columnwidth]{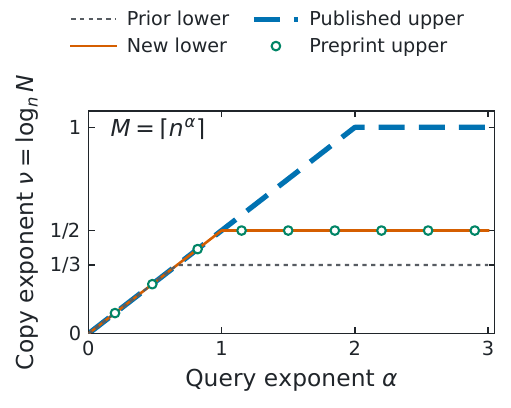}
\caption{Theoretical polynomial exponents at fixed accuracy for $M=\lceil n^\alpha\rceil$. The old lower-bound exponent is $\min\{\alpha/2,1/3\}$, and the new lower-bound exponent is $\min\{\alpha,1\}/2$. The published HLHP upper-bound exponent is $\min\{\alpha/2,1\}$. The upper-bound exponent from Chen et al.'s online shadow-tomography study coincides with the new lower-bound curve; different symbols distinguish them. This exponent plot suppresses logarithmic factors, including the remaining upper-bound factor $\log n$ for fixed $\alpha>1$. All curves are analytic bounds, not fitted data.}
\label{fig:exponent}
\end{figure}

The classical $\Omega(\sqrt M/\varepsilon^2)$ theorem in Ref.~\cite{NSSSU18} assumes both population accuracy and post hoc generalization. The latter also controls the discrepancy between sample and population averages of queries exposed by the interaction. The dimension-independent bounds compared in Refs.~\cite{LT25,COPW26} require only population accuracy and still leave a precision gap; closing it remains open even in that classical comparison. These model distinctions must be retained when comparing the present parity lower bound with lower bounds for general statistical queries.

Precision dependence must still be distinguished from the matching of upper and lower bounds at fixed accuracy. The online protocols in Ref.~\cite{COPW26} and the existing full-Pauli table method \cite{HKP21} give
\begin{equation}
 \begin{split}
 \Ns(n,M,\varepsilon)\le C_1\min\Bigl\{&
 \frac{\sqrt M}{\varepsilon^2},\frac n{\varepsilon^4},\\
 &\frac{\sqrt n}{\varepsilon^3}
 \log\!\left(M+\frac n{\varepsilon^2}+1\right)\Bigr\}.
 \end{split}
 \label{eq:precisionupper}
\end{equation}
The adaptive term in Eq.~\eqref{eq:precisionlower} leaves a precision
gap with these upper bounds.

At fixed accuracy, the identification term and the full-Pauli table method give $\Ns(n,M)=\Theta_\gamma(n)$ when $M\ge2^{\gamma n}$, for fixed $\gamma>0$. For $M=\lceil n^\alpha\rceil$, combining the lower bound with the existing upper bound gives a copy complexity of $\Theta(n^{\alpha/2})$ throughout $0\le\alpha\le1$. For fixed $\alpha>1$, combining the lower bound with the upper bound from Chen et al.'s online shadow-tomography study gives $\Omega(\sqrt n)\le\Ns(n,M)\le O_\alpha(\sqrt n\log n)$. Consequently, the polynomial exponent is
\begin{equation}
 \lim_{n\to\infty}
 \frac{\log\Ns(n,\lceil n^\alpha\rceil)}{\log n}
 =\frac12\min\{\alpha,1\}.
 \label{eq:exponent}
\end{equation}
Figure~\ref{fig:exponent} distinguishes this exponent statement from
the published upper bound.
The logarithmic interval for $\alpha>1$ remains, and intermediate
budgets can retain smaller gaps; for example, $M=\lceil n\log n\rceil$
gives $\Omega(\sqrt n)\le\Ns\le O(\sqrt{n\log n})$.

For diagonal inputs, computational-basis measurement identifies the corresponding classical problem as adaptive parity prediction from $n$-bit samples. HLHP already give a finite-dimensional lower bound for this class; Eq.~\eqref{eq:new} improves its dependence on the number of bits. For general bounded queries, we instantiate Theorem~12 of Ref.~\cite{LT25} on a $2^n$-point universe. Under the source's constant-confidence convention, this gives the fixed-accuracy sample lower bound $\Omega_\alpha(\sqrt n\,\log n)$ for $M=\lceil n^\alpha\rceil$ and each fixed $\alpha>1$ (Supplemental Material, Sec.~S10.1). This rate inherits an unspecified positive constant failure probability from the source. The full adaptive attack uses input obfuscation and removes points identified through previous answers. If each query could be realized by a single parity on $O(n)$ fixed bits while preserving this information restriction and establishing a failure probability strictly greater than $1/3$, it would contradict the task's $2/3$ success guarantee at the corresponding sample scale and thereby close the logarithmic interval. The cited reduction does not provide such a representation or confidence conversion.

The obstruction is present even when computational-basis readout
reveals every diagonal sample without measurement incompatibility.
The mechanism must still support questions selected from its own
earlier predictions of the same preparation.
The theorem concerns arbitrary Pauli weights and adaptive selection
rules; adaptive experimental correlation analysis supplies its
physical setting.
Finite system size already permits this statistical cost to appear at
linear query budgets: $n$ adaptive Pauli predictions on a system
with the same number of qubits require $\Theta(\sqrt n)$ copies, whereas a fixed
list of the same length requires only $O(\log n)$ copies.

\begin{acknowledgments}
\prlheading{Acknowledgments}
This work is supported by Beijing Institute of Technology Research Fund
Program under Grant No. 2024CX01015.
We acknowledge the use of GPT-6 Astra and Opus-5.5 to assist with theoretical calculations.
The authors have checked all mathematical derivations and theorems and
take full responsibility for the conclusions.
\end{acknowledgments}

\prlheading{Data availability}
The analytical constructions and proofs are provided in this article
and its Supplemental Material. The figures show the construction and
analytic bounds; no experimental or simulated data are used.

\makeatletter
\let\pre@bibdata\@empty
\immediate\write\@auxout{\string\citation{pauliBibStyle}}
\makeatother
\bibliographystyle{apsrev4-2}
\bibliography{references}

@CONTROL{pauliBibStyle,
  editor = {1},
  pages = {1},
  title = {0},
  year = {1}
}

@article{Haah17,
  author = {Haah, J. and Harrow, A. W. and Ji, Z. and Wu, X. and Yu, N.},
  title = {Sample-optimal tomography of quantum states},
  journal = {IEEE Trans. Inf. Theory},
  volume = {63},
  pages = {5628},
  year = {2017}
}

@inproceedings{ODonnell16,
  author = {O'Donnell, R. and Wright, J.},
  title = {Efficient quantum tomography},
  booktitle = {Proceedings of STOC 2016},
  publisher = {ACM},
  pages = {899--912},
  year = {2016}
}

@inproceedings{Aaronson18,
  author = {Aaronson, S.},
  title = {Shadow tomography of quantum states},
  booktitle = {Proceedings of STOC 2018},
  publisher = {ACM},
  pages = {325--338},
  year = {2018}
}

@article{HKP20,
  author = {Huang, H.-Y. and Kueng, R. and Preskill, J.},
  title = {Predicting many properties of a quantum system from very few measurements},
  journal = {Nat. Phys.},
  volume = {16},
  pages = {1050},
  year = {2020}
}

@article{HKP21,
  author = {Huang, H.-Y. and Kueng, R. and Preskill, J.},
  title = {{Information-Theoretic Bounds on Quantum Advantage in Machine Learning}},
  journal = {Phys. Rev. Lett.},
  volume = {126},
  pages = {190505},
  year = {2021},
  doi = {10.1103/PhysRevLett.126.190505},
  note = {The fixed-list theorem is Theorem 4 of \href{https://arxiv.org/abs/2101.02464v2}{arXiv:2101.02464v2}}
}

@misc{JHL26,
  author = {Granha Jeronimo, F. and Huang, Q. and Liu, L.},
  title = {{Dimension-Free Polylogarithmic Quantum Shadow Tomography}},
  year = {2026},
  eprint = {2608.06345v3},
  archivePrefix = {arXiv}
}

@article{Semeghini21,
  author = {Semeghini, G. and Levine, H. and Keesling, A. and Ebadi, S. and Wang, T. T. and Bluvstein, D. and Verresen, R. and Pichler, H. and Kalinowski, M. and Samajdar, R. and Omran, A. and Sachdev, S. and Vishwanath, A. and Greiner, M. and Vuleti\'c, V. and Lukin, M. D.},
  title = {Probing topological spin liquids on a programmable quantum simulator},
  journal = {Science},
  volume = {374},
  pages = {1242--1247},
  year = {2021},
  doi = {10.1126/science.abi8794}
}

@article{Laskowski12,
  author = {Laskowski, W. and Richart, D. and Schwemmer, C. and Paterek, T. and Weinfurter, H.},
  title = {{Experimental Schmidt Decomposition and State Independent Entanglement Detection}},
  journal = {Phys. Rev. Lett.},
  volume = {108},
  pages = {240501},
  year = {2012},
  doi = {10.1103/PhysRevLett.108.240501}
}

@inproceedings{Dwork15,
  author = {Dwork, C. and Feldman, V. and Hardt, M. and Pitassi, T. and Reingold, O. and Roth, A.},
  title = {Preserving statistical validity in adaptive data analysis},
  booktitle = {Proceedings of STOC 2015},
  publisher = {ACM},
  pages = {117--126},
  year = {2015}
}

@inproceedings{Hardt14,
  author = {Hardt, M. and Ullman, J.},
  title = {Preventing false discovery in interactive data analysis is hard},
  booktitle = {Proceedings of FOCS 2014},
  publisher = {IEEE},
  pages = {454--463},
  year = {2014}
}

@article{SU15,
  author = {Steinke, T. and Ullman, J.},
  title = {Interactive fingerprinting codes and the hardness of preventing false discovery},
  journal = {Proc. Mach. Learn. Res.},
  volume = {40},
  pages = {1588},
  year = {2015},
  url = {https://proceedings.mlr.press/v40/Steinke15.html},
  note = {The expanded version is \href{https://arxiv.org/abs/1410.1228v2}{arXiv:1410.1228v2}}
}

@article{Bassily21,
  author = {Bassily, R. and Nissim, K. and Smith, A. and Steinke, T. and Stemmer, U. and Ullman, J.},
  title = {Algorithmic stability for adaptive data analysis},
  journal = {SIAM J. Comput.},
  volume = {50},
  pages = {STOC16{\char45}377},
  year = {2021}
}

@inproceedings{Dagan22,
  author = {Dagan, Y. and Kur, G.},
  title = {A bounded-noise mechanism for differential privacy},
  booktitle = {Proceedings of COLT 2022},
  series = {Proc. Mach. Learn. Res.},
  volume = {178},
  pages = {625--661},
  year = {2022}
}

@inproceedings{AaronsonRothblum19,
  author = {Aaronson, S. and Rothblum, G. N.},
  title = {Gentle measurement of quantum states and differential privacy},
  booktitle = {Proceedings of STOC 2019},
  publisher = {ACM},
  pages = {322--333},
  year = {2019}
}

@article{Badescu21,
  author = {B\u{a}descu, C. and O'Donnell, R.},
  title = {Improved quantum data analysis},
  journal = {TheoretiCS},
  volume = {3},
  pages = {7},
  year = {2024},
  doi = {10.46298/theoretics.24.7}
}

@article{HLHP26,
  author = {Huang, J. and Lewis, L. and Huang, H.-Y. and Preskill, J.},
  title = {{Predicting Adaptively Chosen Observables in Quantum Systems}},
  journal = {PRX Quantum},
  volume = {7},
  pages = {010347},
  year = {2026},
  doi = {10.1103/xhn1-vnp9}
}

@misc{COPW26,
  author = {Chen, S. and O'Donnell, R. and Pelecanos, A. and Wright, J.},
  title = {{Online Shadow Tomography Matching the Classical Bounds}},
  year = {2026},
  eprint = {2607.29686v2},
  archivePrefix = {arXiv}
}

@FOOTNOTE{Supplemental,
  note = {See the \hyperref[supplement-start]{Supplemental Material} following these references for complete proofs, the encoding rank bound, and upper-bound parameter conversions; it includes Refs.~\cite{HKP21,SU15,HLHP26,COPW26,NSSSU18,LT25}}
}

@inproceedings{NSSSU18,
  author = {Nissim, K. and Smith, A. and Steinke, T. and Stemmer, U. and Ullman, J.},
  title = {The limits of post-selection generalization},
  booktitle = {Advances in Neural Information Processing Systems},
  volume = {31},
  year = {2018},
  url = {https://proceedings.neurips.cc/paper/2018/hash/77ee3bc58ce560b86c2b59363281e914-Abstract.html},
  note = {The full version is \href{https://arxiv.org/abs/1806.06100v1}{arXiv:1806.06100v1}}
}

@inproceedings{LT25,
  author = {Lyu, X. and Talwar, K.},
  title = {Fingerprinting codes meet geometry: {Improved} lower bounds for private query release and adaptive data analysis},
  booktitle = {Proceedings of STOC 2025},
  publisher = {ACM},
  year = {2025},
  doi = {10.1145/3717823.3718295},
  note = {\href{https://arxiv.org/abs/2412.14396v1}{arXiv:2412.14396v1}}
}

\clearpage
\onecolumngrid
\begingroup
\makeatletter
\let\th@prlstatement\th@plain
\renewcommand\normalsize{%
  \@setfontsize\normalsize{10.95}{13.6}%
  \abovedisplayskip 11pt plus 3pt minus 6pt
  \abovedisplayshortskip 0pt plus 3pt
  \belowdisplayshortskip 6.5pt plus 3.5pt minus 3pt
  \belowdisplayskip \abovedisplayskip
  \let\@listi\@listI}
\renewcommand{\p@subsection}{}
\makeatother
\normalsize
\setcounter{section}{0}
\setcounter{subsection}{0}
\setcounter{equation}{0}
\setcounter{theorem}{0}
\setcounter{lemma}{0}
\setcounter{proposition}{0}
\setcounter{table}{0}
\setcounter{secnumdepth}{2}
\renewcommand{\thesection}{S\arabic{section}}
\renewcommand{\thesubsection}{\thesection.\arabic{subsection}}
\renewcommand{\theequation}{S\arabic{equation}}
\renewcommand{\thetable}{S\arabic{table}}
\providecommand{\theHsection}{}
\providecommand{\theHsubsection}{}
\providecommand{\theHequation}{}
\providecommand{\theHtheorem}{}
\providecommand{\theHlemma}{}
\providecommand{\theHproposition}{}
\providecommand{\theHtable}{}
\renewcommand{\theHsection}{supp.section.\arabic{section}}
\renewcommand{\theHsubsection}{\theHsection.\arabic{subsection}}
\renewcommand{\theHequation}{supp.equation.\arabic{equation}}
\renewcommand{\theHtheorem}{supp.theorem.\arabic{theorem}}
\renewcommand{\theHlemma}{supp.lemma.\arabic{lemma}}
\renewcommand{\theHproposition}{supp.proposition.\arabic{proposition}}
\renewcommand{\theHtable}{supp.table.\arabic{table}}
\phantomsection
\label{supplement-start}
\begin{center}
{\large\bfseries Supplemental Material}\par\medskip
{\bfseries Finite-Size Limits of Adaptive Pauli Prediction}
\end{center}
\medskip

This Supplemental Material gives the complete lower-bound proofs, the
optimality argument for the universal encoding, and the parameter
conversions used for the upper bounds. We restate the task and both
theorems so that all definitions and proof dependencies are local.
Theorem numbers agree with the Letter; equations and sections have
independent supplemental numbering.

\section{Task, notation, and theorem statements}\label{sm:sec:task}
Let $\mathcal P_n=\{I,X,Y,Z\}^{\otimes n}$. A prediction mechanism
receives $N$ copies of an unknown $n$-qubit state $\rho$ before the
interaction begins. In round $j$, a classical analyst submits
$P_j\in\mathcal P_n$ and receives an answer $a_j\in[-1,1]$ before
choosing the next query. The analyst may know $\rho$ and the mechanism's
description, and may use the complete previous transcript. The mechanism
has access to $\rho$ only through its initial copies and the queries it
receives. It cannot read the analyst's private state or a classical
description of $\rho$. Arbitrary joint operations, quantum memory, and
computation are allowed. The parties' private randomness is initially
independent. Repeated queries are allowed.

For a tolerance $\varepsilon>0$, define the simultaneous accuracy event
\begin{equation}
 \Acc_M^{(\varepsilon)}(\rho)
 =\bigcap_{j=1}^M\{|a_j-\Tr(P_j\rho)|\le\varepsilon\}.
 \label{sm:eq:accuracy}
\end{equation}
The copy complexity $\Ns(n,M,\varepsilon)$ is the minimum nonnegative
integer $N$ for which one mechanism satisfies
$\Pr[\Acc_M^{(\varepsilon)}(\rho)]\ge2/3$ for every fixed state and
every legal analyst. The probability includes the parties' randomness
and quantum measurement outcomes. We set $\eps=10^{-3}$ and abbreviate
$\Ns(n,M)=\Ns(n,M,\eps)$. All logarithms have base two unless denoted
by $\ln$. The field $\F$ is the binary field, $[d]=\{1,\ldots,d\}$,
and $\ip{v}{x}$ is the binary inner product. We write
$Z^v=\bigotimes_r Z^{v_r}$, with $Z^0=I$ at each site, and
$e_i^{(d)}$ for the $i$th coordinate vector in $\F^d$.

\begin{theorem}[Fixed accuracy]\label{sm:thm:fixed}
There is a universal constant $c>0$ such that, for all integers
$n\ge2$ and $M\ge1$,
\begin{equation}
 \Ns(n,M)\ge c\sqrt{\min\{n,M\}}.
 \label{sm:eq:new}
\end{equation}
The lower bound already holds for computational-basis diagonal states
and queries in $\{I,Z\}^{\otimes n}$.
\end{theorem}

\begin{theorem}[Precision dependence]\label{sm:thm:precision}
Set $r=\sqrt{\min\{n,M\}}$ and $\ell=\min\{n,\log(M+1)\}$.
There is a universal constant $c_1>0$ such that, for $n\ge2$, $M\ge1$,
and $0<\varepsilon\le1/4$,
\begin{equation}
 \Ns(n,M,\varepsilon)\ge
 c_1\max\{r/\varepsilon,\ell/\varepsilon^2\}.
 \label{sm:eq:precisionlower}
\end{equation}
Both terms hold already for diagonal states and commuting Pauli queries.
\end{theorem}

Theorem~\ref{sm:thm:fixed} follows from four lemmas: the encoding realizes a
Pauli parity, diagonal inputs admit exact classical simulation, that
simulation respects the fingerprinting coalition's information, and
accurate answers imply consistent code responses. For
Theorem~\ref{sm:thm:precision}, a finite-budget dilution argument gives the
adaptive term, while identification among weakly biased diagonal states
gives the second term.

\section{Interactive fingerprinting codes and constants}\label{sm:sec:code}
An interactive fingerprinting code generates a column
$c^j\in\F^d$ in each round $j$. A coalition $S^+\subseteq[d]$ of
exactly $s$ users is selected before the first round, independently of
the generator's private seed. Let $T_{j-1}$ be the set of all users
accused before round $j$, with $T_0=\varnothing$. The coalition receives
only $c_i^j$ for its active members $i\in S^+\setminus T_{j-1}$ and
returns one bit $b_j$. The generator may then accuse additional users;
accused coalition members receive no subsequent code bits. A response
is inconsistent if it agrees with no entry of the full column,
including entries outside the coalition. Define
\begin{equation}
 \theta=\sum_{j=1}^L
 \mathbf1\{b_j\notin\{c_i^j:i\in[d]\}\},\qquad
 \psi=|T_L\setminus S^+|.
 \label{sm:eq:codestatistics}
\end{equation}
Thus $\psi$ counts false accusations, and a column containing both
symbols admits either response.

We use Definition 6 and Theorem 7 in the COLT version of
Steinke--Ullman~\cite{SU15}; these correspond to Definition 2.1 and Theorem 2.2
in Steinke and Ullman's study of interactive fingerprinting codes. For a tolerated inconsistency fraction
$0\le\beta_{\mathrm F}<1/2$ and a false-accusation parameter $\delta_F$, the theorem
bounds the union of the events $\theta\le\beta_{\mathrm F} L$ and
$\psi>\delta_F(d-s)$, with
\begin{align}
 L&=O\!\left(\frac{s^2\log(1/\delta_F)}{(1/2-\beta_{\mathrm F})^4}\right),
 \label{sm:eq:sulength}\\
 \eta&\le\min\{\delta_F(d-s),2^{-\Omega(\delta_F(d-s))}\}
       +\delta_F^{\Omega((1/2-\beta_{\mathrm F})s)}.
 \label{sm:eq:sufailure}
\end{align}
More explicitly, the soundness and completeness proofs in COLT
Theorems 15 and 26 (Theorems 2.8 and 2.19 in that study) give
$\min\{\delta_F(d-s),e^{-\delta_F(d-s)/8}\}$ and
$\delta_F^{(1/2-\beta_{\mathrm F})s/2}$, respectively, when
\mbox{$(1/2-\beta_{\mathrm F})s\ge1$}. Choose
\begin{equation}
 d=2000s,\qquad \beta_{\mathrm F}=1/4,\qquad \delta_F=1/2000.
 \label{sm:eq:fixedcode}
\end{equation}
This strictly positive choice lies in the construction's parameter
range. Since $\{\theta=0\}\subseteq\{\theta\le L/4\}$, the same
guarantee bounds perfect consistency. There is a universal $C_F>0$ such that
\begin{equation}
 \begin{split}
 L&\le C_Fs^2,\\
 \Pr[\theta=0\ \text{or}\ \psi>(d-s)/2000]
 &\le\eta_s
 \le e^{-1999s/16000}+2000^{-s/8}<1/3
 \quad(s\ge s_0),
 \end{split}
 \label{sm:eq:codeparameters}
\end{equation}
where one may take $s_0=10$. The completeness condition is $s/4\ge1$.
Both terms decrease with $s$, and their sum at $s=10$ is less than
$0.287$. Relabeling the code alphabet
$\{-1,1\}$ as $\{0,1\}$ preserves consistency. The guarantee applies
to every legal coalition strategy, with unrestricted internal
computation.

\section{Additive parity encoding and its universal optimum}\label{sm:sec:encoding}
The code must be presented as Pauli queries while hiding the current
code bits of users absent from the sample. Before the interaction,
choose a uniform binary matrix $K\in\F^{d\times L}$ and form
\begin{equation}
 x_i=(e_i^{(d)},K_{i,*}),\qquad
 \rho_K=\frac1d\sum_{i=1}^d\ket{x_i}\bra{x_i}.
 \label{sm:eq:state}
\end{equation}
Each label contains $d+L$ bits. The analyst knows the entire key matrix;
the mechanism receives only copies of the fixed state $\rho_K$.
In round $j$, define the masked plaintext and its encryption by
\begin{equation}
 q_i^j=\begin{cases}c_i^j,&i\notin T_{j-1},\\0,&i\in T_{j-1},\end{cases}
 \qquad u^j=q^j\oplus K_{*,j},\qquad
 v_j=(u^j,e_j^{(L)}),\qquad P_j=Z^{v_j}.
 \label{sm:eq:query}
\end{equation}
The analyst sends $P_j$ to the mechanism and turns its answer into
$b_j=\mathbf1\{a_j<0\}$. The code generator receives this bit alone.
Setting the plaintext to zero at accused positions ensures that later
queries can be generated without those users' new code bits.

\begin{lemma}[Parity realization]\label{sm:lem:parity}
For every key matrix, code column, and accusation set, the labels and
queries in Eqs.~\eqref{sm:eq:state} and \eqref{sm:eq:query} satisfy
\begin{equation}
 \ip{v_j}{x_i}=u_i^j\oplus K_{ij}=q_i^j.
 \label{sm:eq:decode}
\end{equation}
Consequently, each query is a single Pauli observable and has mean
\begin{equation}
 \mu_j:=\Tr(P_j\rho_K)=\frac1d\sum_{i=1}^d(-1)^{q_i^j}.
 \label{sm:eq:mean}
\end{equation}
\end{lemma}
\begin{proof}
The first block of the inner product selects $u_i^j$, because
$\ip{u^j}{e_i^{(d)}}=u_i^j$. The second block selects $K_{ij}$,
because $\ip{e_j^{(L)}}{K_{i,*}}=K_{ij}$. Their sum in $\F$ is
$u_i^j\oplus K_{ij}=q_i^j$. Since
$Z^v\ket{x}=(-1)^{\ip{v}{x}}\ket{x}$, averaging these eigenvalues
in the mixture \eqref{sm:eq:state} proves Eq.~\eqref{sm:eq:mean}.
\end{proof}

The Pauli construction in Ref.~\cite{HLHP26} uses
$\lceil\log_2d\rceil+2dL$ qubits. Here the number of qubits per copy
is $d+L$; the analyst still stores the $dL$-bit key matrix.
Steinke--Ullman already use data $(i,sk_i)$ and obtain
$O(s^2)$-bit information-theoretic hard instances for unrestricted
statistical queries with sufficiently long keys~\cite{SU15}.
The independent round pads instantiate that existing key framework.
The present representation makes each masked query a single parity.
The joint-law argument in Sec.~\ref{sm:sec:simulation} verifies this
representation within the same fingerprinting framework.

To characterize the size required by this interface, consider maps
$E:[d]\times\F^L\to\F^m$ and
$V:\F^d\times[L]\to\F^m$ satisfying
\begin{equation}
 \ip{V(u,j)}{E(i,k)}=u_i\oplus k_j
 \quad\text{for every }i\in[d],\ k\in\F^L,\ u\in\F^d,\ j\in[L].
 \label{sm:eq:universal}
\end{equation}
The row encoder sees only an identifier and its key row; the query
encoder sees only a ciphertext column and its round. The maps may be
nonlinear, and the identity must hold before any key matrix is sampled.

\begin{proposition}[Universal exact encoding]\label{sm:prop:encoding}
For $d,L\ge1$, the minimum $m$ in Eq.~\eqref{sm:eq:universal} is
$d+L$. If a query-dependent sign correction is allowed, the minimum
is $d+L-1$.
\end{proposition}
\begin{proof}
Regard the coordinates $E_1,\ldots,E_m$ as binary-valued functions on
$[d]\times\F^L$. Their linear span contains
$f_{u,j}(i,k)=u_i\oplus k_j$, hence contains
\begin{equation}
 \zeta_j(i,k)=k_j=f_{0,j}(i,k),\qquad
 \iota_r(i,k)=\mathbf1\{i=r\}
 =f_{e_r,j}(i,k)\oplus f_{0,j}(i,k).
 \label{sm:eq:rankfunctions}
\end{equation}
These $d+L$ functions are linearly independent over $\F$.
Indeed, in $\sum_r\lambda_r\iota_r+\sum_j\xi_j\zeta_j=0$, setting $k=0$
and then $i=r$ gives each $\lambda_r=0$; varying $k$ then gives each
$\xi_j=0$. Their span therefore has dimension $d+L$, so $m\ge d+L$.
Equality is attained by $E(i,k)=(e_i^{(d)},k)$ and
$V(u,j)=(u,e_j^{(L)})$.

A query-dependent sign $(-1)^{h(u,j)}$ changes the required binary
identity to $\ip{V(u,j)}{E(i,k)}\oplus h(u,j)=u_i\oplus k_j$.
The span of $1,E_1,\ldots,E_m$ contains the same $d+L$ independent
functions, giving $m+1\ge d+L$. To attain equality, omit $\iota_d$ and use
\begin{equation}
 \begin{split}
 E(i,k)&=(\iota_1(i,k),\ldots,\iota_{d-1}(i,k),k),\\
 V(u,j)&=(u_1\oplus u_d,\ldots,u_{d-1}\oplus u_d,e_j^{(L)}),
 \qquad h(u,j)=u_d.
 \end{split}
 \label{sm:eq:signedencoding}
\end{equation}
For $i<d$, the offset cancels the extra $u_d$; for $i=d$, the
identifier block is zero and the offset supplies $u_d$. This also
covers $d=1$, when that block is empty.
\end{proof}

Proposition~\ref{sm:prop:encoding} gives the exact optimum for the
universal interface \eqref{sm:eq:universal}. An encoding specialized
to an already fixed key matrix has different quantifiers; the rank
argument uses all possible key rows.

\section{Exact simulation of the adaptive interaction}\label{sm:sec:simulation}
The parity identity gives the correct mean. Applying the fingerprinting
code additionally requires that the mechanism's view can be produced
using only a legal coalition's information. We establish this for the
complete adaptive record.

\begin{lemma}[Diagonal simulation]\label{sm:lem:diagonal}
For a computational-basis diagonal state, an arbitrary quantum mechanism
can be simulated from independent classical basis labels with exactly
the same complete adaptive transcript law.
\end{lemma}
\begin{proof}
Write $\rho=\sum_xp_x\ket{x}\bra{x}$. Draw $s$ independent labels
$X_1,\ldots,X_s$ from $p$ and prepare their basis states. Their
averaged input is
\begin{equation}
 \mathbb E\!\left[\bigotimes_{k=1}^s\ket{X_k}\bra{X_k}\right]
 =\rho^{\otimes s}.
 \label{sm:eq:diagonalinput}
\end{equation}
For each ordered pure-input tuple, a simulator uses the original
mechanism's causal conditional output kernels. It needs only the tuple,
the mechanism's description, and queries already received. Linearity of
quantum instruments implies equality of the probabilities of all
measurable transcript events after averaging over the tuple. This
remains true when each next query is chosen using earlier answers:
conditioned on the same past transcript, the analyst uses the same
next-query kernel in the two experiments. Induction over rounds gives
the complete transcript law.

Answers lie in the standard Borel space $[-1,1]$, so independent
uniform random seeds can realize these kernels. Arbitrary computation
is permitted. A description of $p$ is not an additional simulator input.
For $\rho_K$, the labels are $x_{J_1},\ldots,x_{J_s}$ with
$J_1,\ldots,J_s$ independent and uniform on $[d]$.
\end{proof}

Retain the ordered index tuple $J=(J_1,\ldots,J_s)$, including
repetitions, and let $S=\{J_1,\ldots,J_s\}$. Pad $S$ deterministically
to a set $S^+$ of exactly $s$ users. A legal coalition samples the rows
$K_{i,*}$ for $i\in S$ and runs the simulator on the ordered labels.
For an active sampled user it sets $u_i^j=c_i^j\oplus K_{ij}$;
for an accused sampled user it sets $u_i^j=K_{ij}$. At every
unsampled position it uses a fresh independent fair bit. It submits
$P_j=Z^{(u^j,e_j^{(L)})}$ to the simulator and returns
$b_j=\mathbf1\{a_j<0\}$ to the generator.

\begin{lemma}[Legal coalition]\label{sm:lem:coalition}
This coalition strategy and the real attack have the same joint law
of the sample indices, padded coalition, full code columns, accusations,
queries, answers, and response bits.
\end{lemma}
\begin{proof}
Let $R_F,R_A$ be complete independent random seeds for the code
generator and simulator. In the real experiment, condition on
\begin{equation}
 \Xi=(J,K_S,R_F,R_A),\qquad \bar S=[d]\setminus S.
 \label{sm:eq:initial}
\end{equation}
The sampled rows $K_S$ include all their future entries. The remaining
matrix $K_{\bar S,*}$ is still independent and uniform. The coalition
experiment has the same distribution of $\Xi$.

Before round $j$, consider the proof filtration
\begin{equation}
 \mathcal H_{j-1}
 =\sigma(\Xi,K_{\bar S,1},\ldots,K_{\bar S,j-1}).
 \label{sm:eq:filtration}
\end{equation}
All past queries and answers are measurable with respect to it: the
simulator started from $K_S$ and has used only earlier queries.
The code generator uses $R_F$ and earlier response bits, so the
current $c^j$, $T_{j-1}$, and $q^j$ are also measurable. The unused
column $K_{\bar S,j}$ is independent of $\mathcal H_{j-1}$.
Consequently, for each $y\in\F^{d-|S|}$,
\begin{equation}
 \Pr[u^j_{\bar S}=y\mid\mathcal H_{j-1}]
 =\Pr[q^j_{\bar S}\oplus K_{\bar S,j}=y\mid\mathcal H_{j-1}]
 =2^{-(d-|S|)}.
 \label{sm:eq:uniform}
\end{equation}
Conditioning only on $\Xi$ and the past joint record leaves this law
uniform. Thus the unsampled ciphertext has exactly the law of the
fresh bits used by the coalition. At sampled positions the real and
simulated ciphertexts agree by construction, including at accused
positions.

Identical simulator and generator kernels now yield identical
conditional laws for the next answer, response bit, and accusations.
Starting at the empty history, induction and averaging over $\Xi$ give
\begin{equation}
 (J,S^+,\mathsf{Rec}_L)_{\rm real}
 \overset{\rm law}{=}(J,S^+,\mathsf{Rec}_L)_{\rm coalition},
 \qquad \mathsf{Rec}_L=(c^j,I_j,P_j,a_j,b_j)_{j=1}^L,
 \label{sm:eq:jointlaw}
\end{equation}
where $I_j=T_j\setminus T_{j-1}$ is the newly accused set. Since
$T_0=\varnothing$, the $I_j$ determine every $T_j$. Including $S^+$
therefore transfers both $\theta$ and $\psi=|T_L\setminus S^+|$.

For legality, the coalition was selected independently of $R_F$.
It receives code bits only for active members of $S^+$, ignores those
in $S^+\setminus S$, and uses known key bits at accused sampled
positions. It needs removal information only for its own members.
The filtration and the full code columns in the record are proof
variables, not extra information supplied to the coalition.
Conditioning on a fixed coalition and its private seed gives a finite
response table on the finite code histories. Averaging such tables
covers arbitrary internal computation and real-valued simulator
outputs, so the strategy falls within the fingerprinting guarantee.
\end{proof}

\section{Consistency, fixed-state quantifiers, and dimension fitting}\label{sm:sec:fixedproof}
We first use accurate means to control the code responses. The argument
works at tolerance $1/2$, which will also permit dilution.

\begin{lemma}[Accuracy implies consistency]\label{sm:lem:consistency}
For $0\le\tau\le1/2$, the attack satisfies
\begin{equation}
 \Acc_L^{(\tau)}(\rho_K)
 \subseteq\{\theta=0\}\cup\{\psi>(d-s)/2000\}.
 \label{sm:eq:containment}
\end{equation}
\end{lemma}
\begin{proof}
Suppose $\psi\le(d-s)/2000$. The accusation sets are increasing,
and at most $s$ accused users belong to $S^+$. In every round,
\begin{equation}
 \frac{|T_{j-1}|}{d}
 \le\frac{s+(d-s)/2000}{d}
 =\frac{3999}{4\,000\,000}.
 \label{sm:eq:accused}
\end{equation}
If $c^j$ contains both symbols, either response is consistent.
If it is all zero, then $q^j$ is all zero and $\mu_j=1$.
If it is all one, masking only the accused entries gives
\begin{equation}
 \mu_j=-1+2|T_{j-1}|/d\le-0.9980005.
 \label{sm:eq:allone}
\end{equation}
An answer accurate to $\tau\le1/2$ is positive in the all-zero case
and negative in the all-one case. Therefore
$b_j=\mathbf1\{a_j<0\}$ agrees with the full column in both cases.
On simultaneous accuracy no round is inconsistent, so $\theta=0$.
This proves the containment.
\end{proof}

\subsection{A fixed hard state for each mechanism}\label{sm:sec:quantifiers}
The event on the right of Eq.~\eqref{sm:eq:containment} is determined
by the joint record in Lemma~\ref{sm:lem:coalition}. Its probability
is at most $\eta_s$ by Eq.~\eqref{sm:eq:codeparameters}. Thus, for every
$s$-copy mechanism and every $0\le\tau\le1/2$,
\begin{equation}
 \mathbb E_K\Pr[\Acc_L^{(\tau)}(\rho_K)\mid K]\le\eta_s<1/3.
 \label{sm:eq:averaging}
\end{equation}
Only the enclosing fingerprinting event is transferred; the coalition
need not construct the unsampled keys or their state. Equation
\eqref{sm:eq:averaging} implies that some fixed $K$ has success probability
at most $\eta_s$, in particular below $2/3$. The quantifier order is
\begin{equation}
 \forall\mathsf A\quad\exists K\quad\exists\mathsf{Adv}_K,
 \label{sm:eq:quantifiers}
\end{equation}
where $\mathsf A$ is the candidate mechanism and $\mathsf{Adv}_K$ is the
code analyst using that key. The state and analyst are chosen before
the mechanism's run. Averaging over keys is a proof step, and does not
make the input state change during the interaction.

\subsection{Proof of Theorem 1 for every parameter size}
\begin{proof}[Proof of Theorem~\ref{sm:thm:fixed}]
Set $B=2000+C_F$. Since $s\ge1$, the attack uses
\begin{equation}
 d+L\le2000s+C_Fs^2\le Bs^2,\qquad L\le Bs^2.
 \label{sm:eq:budget}
\end{equation}
Let $w=\min\{n,M\}$. If $w\ge4Bs_0^2$, choose
\begin{equation}
 s=\left\lfloor\sqrt{\frac{w}{B}}\right\rfloor.
 \label{sm:eq:floorchoice}
\end{equation}
Writing $x=\sqrt{\frac{w}{B}}$, the assumed threshold gives $x\ge2s_0\ge2$.
The elementary inequalities $x-1\le\lfloor x\rfloor\le x$ imply
\begin{equation}
 s\ge x-1\ge x/2\ge s_0,\qquad
 Bs^2\le w=\min\{n,M\},\qquad
 s\ge\frac{\sqrt w}{2\sqrt B}.
 \label{sm:eq:floorbounds}
\end{equation}
Thus both the qubit count and the query count fit. Fix unused qubits
in $\ket0$ and extend each attack query by identities. Pad the remaining
query rounds by identities. A universal guarantee for all $M$ rounds
would imply accuracy on the $L$-round attack prefix, even when the
mechanism knows $n$ and $M$ in advance. Equation~\eqref{sm:eq:averaging}
at $\tau=\eps$ rules out an $s$-copy mechanism. It also rules out every
smaller budget, since an $s$-copy mechanism could discard extra copies.
Consequently,
\begin{equation}
 \Ns(n,M)>s\ge\frac{\sqrt w}{2\sqrt B}.
 \label{sm:eq:fixedlarge}
\end{equation}

If $1\le w<4Bs_0^2$, at least one copy is necessary. Submit the same
first query $Z_1$ for the states $\ket0^{\otimes n}$ and
$\ket1\otimes\ket0^{\otimes(n-1)}$. Their means are $1$ and $-1$,
whose $\eps$-accuracy intervals are disjoint. A zero-copy mechanism
has the same first-answer distribution in both experiments; the two
disjoint success events cannot both have probability at least $2/3$.
Hence $\Ns(n,M)\ge1$. Since
$\sqrt w<2\sqrt B s_0$, this implies
$\Ns(n,M)\ge\sqrt w/(2\sqrt B s_0)$. Combining the two ranges proves
the theorem with $c=1/(2\sqrt B s_0)$.
\end{proof}

\section{Precision-dependent identification}\label{sm:sec:identification}
A fixed query list supplies a complementary lower bound whose precision
dependence also handles the small sizes in Theorem~\ref{sm:thm:precision}.
Set
\begin{equation}
 Q_{\mathrm{id}}=\min\{M,2^n-1\},\qquad R=Q_{\mathrm{id}}+1,\qquad \Delta=4\varepsilon\le1.
 \label{sm:eq:identparameters}
\end{equation}
Choose distinct nonzero $v_1,\ldots,v_{Q_{\mathrm{id}}}\in\F^n$ and define
\begin{equation}
 \omega_0=I/2^n,\qquad
 \omega_i=(I+\Delta Z^{v_i})/2^n\quad(1\le i\le Q_{\mathrm{id}}).
 \label{sm:eq:identstates}
\end{equation}
These states are positive and normalized: the eigenvalues are
$2^{-n}(1\pm \Delta)$ and each nonidentity parity has zero trace.
Character orthogonality gives
\begin{equation}
 \Tr(Z^{v_j}\omega_i)=\Delta\delta_{ij}\quad(i\ge1),\qquad
 \Tr(Z^{v_j}\omega_0)=0.
 \label{sm:eq:identmeans}
\end{equation}
The vectors need not be linearly independent. Indeed,
$2^{-n}\sum_x(-1)^{\ip{v_i\oplus v_j}{x}}$ is one if $i=j$ and zero
otherwise, which is exactly the orthogonality needed here.

Use the same fixed list $Z^{v_1},\ldots,Z^{v_{Q_{\mathrm{id}}}}$ for every label, and
pad by identities to reach $M$ queries. Decode label $i$ if its answer
is the unique one exceeding $2\varepsilon$, and decode zero otherwise.
On simultaneous accuracy, every zero-mean answer is at most
$\varepsilon$, while the one nonzero-mean answer, when present, is at
least $3\varepsilon$. Therefore each label is decoded correctly with
probability at least $2/3$.

\subsection{Information per sample}
Let $p_i$ be the basis distribution of $\omega_i$; $p_0$ is uniform.
For $i\ge1$,
$p_i(x)/p_0(x)=1+\Delta(-1)^{\ip{v_i}{x}}$. Relative entropy in bits obeys
\begin{align}
 D_2(p_i\Vert p_0)
 &=\frac1{\ln2}\sum_xp_i(x)\ln\frac{p_i(x)}{p_0(x)}\nonumber\\
 &\le\frac1{\ln2}\sum_xp_i(x)
       \left(\frac{p_i(x)}{p_0(x)}-1\right)
 =\frac1{\ln2}\sum_x\frac{(p_i(x)-p_0(x))^2}{p_0(x)}
 =\frac{\Delta^2}{\ln2}.
 \label{sm:eq:chi}
\end{align}
The inequality uses $\ln z\le z-1$. The middle equality follows
from $\sum_x(p_i(x)-p_0(x))=0$, and the last uses the displayed likelihood
ratio. At the endpoint $\Delta=1$, zero-probability terms are interpreted as
$0\log0=0$; the bound follows by continuity and $p_0$ has full support.

Let the label $V$ be uniform on $\{0,\ldots,Q_{\mathrm{id}}\}$, and let $Y$ be the
entire answer record. By Lemma~\ref{sm:lem:diagonal}, $Y$ is a channel from
$N$ conditionally independent samples $X^N$, with no other input
depending on the label. This last condition uses the same fixed query
list for all labels. Data processing and additivity give
\begin{align}
 I(V;Y)&\le I(V;X^N)\nonumber\\
 &\le\frac1R\sum_{i=0}^{Q_{\mathrm{id}}}
 D_2(p_i^{\otimes N}\Vert p_0^{\otimes N})
 \le\frac{N \Delta^2}{\ln2}.
 \label{sm:eq:information}
\end{align}
For completeness, let $\bar p_N=R^{-1}\sum_i p_i^{\otimes N}$.
Expanding the logarithm into a ratio relative to $\bar p_N$ and a
ratio of $\bar p_N$ to $p_0^{\otimes N}$ yields the exact identity
\begin{equation}
 \frac1R\sum_iD_2(p_i^{\otimes N}\Vert p_0^{\otimes N})
 =I(V;X^N)+D_2(\bar p_N\Vert p_0^{\otimes N}).
 \label{sm:eq:informationidentity}
\end{equation}
Nonnegativity of the second term proves the middle inequality in
Eq.~\eqref{sm:eq:information}. Joint quantum measurements are included
in the channel from $X^N$ to $Y$.

\subsection{Fano bound and the binary cases}
Let $\widehat V(Y)$ be the decoder and
$E=\mathbf1\{\widehat V(Y)\ne V\}$. When $E=0$, the label is known
from $Y$; when $E=1$, it has at most $R-1$ possible values. Since
$E$ is determined by $(V,Y)$, the conditional entropy chain rule gives,
for $p_e=\Pr(E=1)\le1/3$,
\begin{equation}
 H(V\mid Y)
 \le h_2(p_e)+p_e\log(R-1)
 \le1+\tfrac13\log R,
 \label{sm:eq:fano}
\end{equation}
where $h_2$ is binary entropy. If $R\ge8$, then $\log R\ge3$ and
\begin{equation}
 I(V;Y)\ge\tfrac23\log R-1\ge\tfrac13\log R,
 \qquad
 N\ge\frac{(\ln2)\log R}{48\varepsilon^2}.
 \label{sm:eq:identlarge}
\end{equation}
Here we substituted $\Delta^2=16\varepsilon^2$ into
Eq.~\eqref{sm:eq:information}.

For $2\le R<8$, compare labels zero and one. The event that the first
answer exceeds $2\varepsilon$ has probability $\pi_0\le1/3$ for
label zero and $\pi_1\ge2/3$ for label one. Applying data processing
to this event gives
\begin{equation}
 \begin{split}
 \frac{N \Delta^2}{\ln2}
 &\ge D_2(p_1^{\otimes N}\Vert p_0^{\otimes N})
 \ge D_2(\operatorname{Bern}(\pi_1)\Vert
          \operatorname{Bern}(\pi_0))\\
 &\ge D_2(\operatorname{Bern}(2/3)\Vert
          \operatorname{Bern}(1/3))=\tfrac13.
 \end{split}
 \label{sm:eq:binary}
\end{equation}
For $0<\pi_0<\pi_1<1$,
$D_2(\operatorname{Bern}(\pi_1)\Vert\operatorname{Bern}(\pi_0))$
increases in $\pi_1$ and decreases in $\pi_0$: its derivatives are
$\log[\pi_1(1-\pi_0)/(\pi_0(1-\pi_1))]>0$ and
$(\pi_0-\pi_1)/(\pi_0(1-\pi_0)\ln2)<0$.
The boundary cases follow by continuity, including infinite divergence.
Thus $N\ge\ln2/(48\varepsilon^2)>1/(72\varepsilon^2)$.
Since $\log R<3$ here, both ranges imply
\begin{equation}
 \Ns(n,M,\varepsilon)\ge\frac{\log R}{216\varepsilon^2}
 =\frac{\min\{n,\log(M+1)\}}{216\varepsilon^2}.
 \label{sm:eq:precisionidentify}
\end{equation}
The equality uses $R=\min\{M+1,2^n\}$. In particular, $n\ge2$ and
$M\ge1$ imply $\log R\ge1$, so this argument includes a single query.

At fixed accuracy, use the same construction with
$\Delta=4\eps$. Substituting $\varepsilon=\eps=10^{-3}$ into
Eq.~\eqref{sm:eq:precisionidentify} yields, in particular,
\begin{equation}
 \Ns(n,M)\ge\tfrac13\min\{n,\log(M+1)\}.
 \label{sm:eq:identify}
\end{equation}
This is a finite-list form of the high-entropy Pauli identification
method used in Ref.~\cite{HKP21}; the weak-bias construction above
provides a single proof for the full precision range.

\section{Finite-budget dilution and proof of Theorem 2}\label{sm:sec:dilution}
For the adaptive term, pad the state in Eq.~\eqref{sm:eq:state} to
$n$ qubits as in Sec.~\ref{sm:sec:fixedproof} and define
\begin{equation}
 \sigma_{K,p}=(1-p)\ket{0^n}\bra{0^n}+p\rho_K,
 \qquad p=2\varepsilon.
 \label{sm:eq:dilution}
\end{equation}
Every attack query is a $Z$ parity with eigenvalue $+1$ on
$\ket{0^n}$, so its diluted mean is $1-p+p\mu_j$.
The diluted analyst forwards the same queries and decodes each
mechanism answer $A_j$ by
\begin{equation}
 a_j=\Pi_{[-1,1]}\!\left(\frac{A_j-(1-p)}p\right),
 \label{sm:eq:rescale}
\end{equation}
where $\Pi$ clips to the indicated interval. Since $\mu_j\in[-1,1]$,
clipping cannot increase its error. An $\varepsilon$-accurate diluted
answer therefore gives an answer accurate to
$\varepsilon/p=1/2$ for $\mu_j$, as required by
Lemma~\ref{sm:lem:consistency}. The analyst uses only its key, public $p$,
and past answers, and passes $b_j=\mathbf1\{a_j<0\}$ to the code.

\subsection{A simulation with exactly \texorpdfstring{$s$}{s} hard-state copies}
Let $\mathsf A$ be a candidate mechanism using $N$ diluted copies.
To simulate it from $s$ copies of $\rho_K$, first measure these diagonal
copies in the computational basis, obtaining $X_1,\ldots,X_s$.
Before any query, draw independent flags
$F_1,\ldots,F_N\sim\operatorname{Bern}(p)$ and put
$N_{\mathrm h}=\sum_{r=1}^NF_r$. If $N_{\mathrm h}>s$, abort the internal simulation and
answer zero causally thereafter. Otherwise prepare the entire
$N$-register input of $\mathsf A$: use $\ket{0^n}$ when $F_r=0$,
and the next unused $\ket{X_k}$ when $F_r=1$. Run its instruments
with independent private randomness and return the rescaled answer
\eqref{sm:eq:rescale} at each step. This is a legal $s$-copy mechanism;
it is not given $K$ or a description of the sample distribution.

To compare the full interactions, extend the $X_k$ to an infinite
independent stream only in an ideal experiment. With the same flags
and private seeds, this experiment prepares exactly
$\sigma_{K,p}^{\otimes N}$ after averaging. On $N_{\mathrm h}\le s$, the ideal
and finite-budget inputs agree before the first query. If their past
records agree, the analyst selects the same next query and the
mechanism applies the same instrument with the same randomness;
therefore their next answers also agree. By induction, the complete
adaptive records agree except with probability at most $\Pr[N_{\mathrm h}>s]$.
This comparison uses the original unconditioned diluted input law.

Full accuracy implies accuracy of the $L$-round prefix. Apply
Eq.~\eqref{sm:eq:averaging} at $\tau=1/2$ to the finite-budget simulator
and then use the coupling. Since $N_{\mathrm h}$ is a nonnegative integer with
mean $Np$, Markov's inequality gives
\begin{equation}
 \mathbb E_K\Pr[\Acc_M^{(\varepsilon)}(\sigma_{K,p})\mid K]
 \le\eta_s+\Pr[N_{\mathrm h}>s]
 \le\eta_s+\frac{Np}{s+1}.
 \label{sm:eq:truncation}
\end{equation}
If $Np\le s/6$, the right side is strictly below
$1/3+1/6=1/2$. Some fixed key then defeats $\mathsf A$.
Consequently every universally valid mechanism must satisfy
\begin{equation}
 N>\frac{s}{12\varepsilon}.
 \label{sm:eq:dilutionbudget}
\end{equation}
The key, state, and diluted analyst are fixed before the mechanism's
run, with the same quantifier order as Eq.~\eqref{sm:eq:quantifiers}.

\subsection{Uniform constants, including small sizes}
\begin{proof}[Proof of Theorem~\ref{sm:thm:precision}]
Let $w=\min\{n,M\}$ and use $B$ and $s_0$ from
Sec.~\ref{sm:sec:fixedproof}. For $w\ge4Bs_0^2$, set
\begin{equation}
 s=\left\lfloor\sqrt{\frac{w}{B}}\right\rfloor.
 \label{sm:eq:precisionfloor}
\end{equation}
The full inequalities in Eq.~\eqref{sm:eq:floorbounds} give $s\ge s_0$,
$Bs^2\le w$, and $s\ge\sqrt w/(2\sqrt B)$. Thus the diluted attack
fits and Eq.~\eqref{sm:eq:dilutionbudget} yields
\begin{equation}
 N>\frac{\sqrt w}{24\sqrt B\varepsilon}.
 \label{sm:eq:precisionlarge}
\end{equation}
For $1\le w<4Bs_0^2$, Eq.~\eqref{sm:eq:precisionidentify} and
$\varepsilon\le1/4$ imply
\begin{equation}
 N\ge\frac1{216\varepsilon^2}
 \ge\frac1{54\varepsilon}
 \ge\frac{\sqrt w}{108\sqrt B s_0\varepsilon}.
 \label{sm:eq:precisionsmall}
\end{equation}
Combining both ranges gives the adaptive lower term uniformly.
Equation~\eqref{sm:eq:precisionidentify} gives the identification term
for the same universal task. The two hard-state families may differ:
every universally valid mechanism must satisfy both bounds, hence
their maximum. Taking
\begin{equation}
 c_1=\min\left\{\frac1{108\sqrt B s_0},\frac1{216}\right\}
 \label{sm:eq:precisionconstant}
\end{equation}
proves Eq.~\eqref{sm:eq:precisionlower}. All states used are diagonal and
all queries are $Z$ parities or identities.
\end{proof}

\section{Upper bounds and parameter conversions}\label{sm:sec:upper}
We give the conversions separately because they use two accuracy
conventions and draw on both published results and Chen et al.'s
online shadow-tomography study.
Proposition C2 of the published Ref.~\cite{HLHP26}
(Appendix C, Sec.~2, p.~010347-32), with the copy bound in Eq.~(C16),
applies for $\varepsilon,\delta\in(0,1/2)$ under a large-query
condition. For $M\ge2^{22}$, we use the explicitly parenthesized
sufficient condition
\begin{equation}
 \varepsilon\delta\ge
 4\exp\!\left[-\frac{M}{(\log M)^2(\log\log M)^4}\right].
 \label{sm:eq:hlhpcondition}
\end{equation}
The published expression uses compact iterated-log notation.
Even if its last factor is read literally as $\log[(\log M)^4]$,
Eq.~\eqref{sm:eq:hlhpcondition} suffices: writing $z=\log\log M>2$
gives $z^4\ge4z=\log[(\log M)^4]$ for either natural or binary
logarithms in this range. The larger denominator makes our condition
stronger.

The resulting bound is
$O(\varepsilon^{-4}\min\{\sqrt{M\log(1/(\varepsilon\delta))},
n\log(1/\delta)\})$ copies. At $\varepsilon=\eps$ and $\delta=1/3$,
a sufficient threshold for Eq.~\eqref{sm:eq:hlhpcondition} is
$M_0=2^{22}$, for either natural or binary logarithms in that condition.
To verify the binary case, put $x=\log_2 M$ and
$\Phi(M)=M/[(\log_2 M)^2(\log_2\log_2 M)^4]$.
For $x\ge22$,
\begin{equation}
 \begin{split}
 \frac{d}{dx}\ln \Phi(2^x)
 &=\ln2-\frac2x-\frac4{x\ln x}
 >\frac12-\frac1{11}-\frac2{33}>0,\\
 \Phi(M_0)&>\frac{2^{22}}{22^2\,5^4}>13>\ln(12000).
 \end{split}
 \label{sm:eq:hlhpthreshold}
\end{equation}
Here $\ln22>3$ and $\log_2 22<5$.
It follows that $4e^{-\Phi(M)}<1/3000=\eps/3$ for every $M\ge M_0$.
The positive logarithmic factors in the denominator are smaller with
natural logarithms on this range, so the same threshold suffices.
This is a conservative threshold, not the first admissible value.
Thus the $O(\sqrt M)$ branch holds at fixed accuracy for all
$M\ge M_0$. The remaining $1\le M<M_0$ are covered by direct
measurements, with a larger universal constant:
use a fresh batch for each round and a conditional failure budget of
$1/(3M)$. Concentration for the $\pm1$ outcomes requires
$O(\eps^{-2}\log(M+1))$ copies per query; a union bound remains valid
for the adaptive queries because each fresh batch is independent
of the previous record. For bounded $M$, this is $O(1)$ and hence
$O(\sqrt M)$.

Theorem 4 in Huang, Kueng, and Preskill's study of quantum advantage
in machine learning (Ref.~\cite{HKP21}) allows a fixed Pauli list of size $Q$
to be estimated using
$O(\varepsilon^{-4}\log(Q/\delta))$ copies. The protocol first uses
two-copy Bell measurements to estimate squared Pauli magnitudes, then
uses a separate quantum memory for gentle collective sign estimation
(Appendix E.2, including Lemma 5). The latter step supplies signed
expectations; Bell magnitudes alone do not. Apply the theorem to the full
list of $Q=4^n$ Paulis and choose $\delta=1/10$. This requires
$O(n/\varepsilon^4)$ copies. Its one simultaneous success event
controls the entire table; every subsequent adaptive query can be
answered by looking up its entry, for arbitrary $M$. At fixed
$\varepsilon=\eps$, this and the published Pauli protocol give
\begin{equation}
 \Ns(n,M)\le C_0\min\{\sqrt M,n\}.
 \label{sm:eq:published}
\end{equation}
For a fixed list of only $M$ Pauli queries, the same theorem instead
gives $O(\log(M+1))$ copies at fixed accuracy and confidence.

For the online protocols in Chen et al.'s online shadow-tomography study~\cite{COPW26}, convert
a Pauli query to an effect and set the effect accuracy to
\begin{equation}
 E_j=(I+P_j)/2,\qquad \varepsilon_{\mathrm{eff}}=\varepsilon/2.
 \label{sm:eq:effectconversion}
\end{equation}
If the protocol's estimate is $\widehat e_j$, return
$2\Pi_{[0,1]}(\widehat e_j)-1$. Since
$\Tr(P_j\rho)=2\Tr(E_j\rho)-1$, an $\varepsilon_{\mathrm{eff}}$-accurate estimate yields
an $\varepsilon$-accurate Pauli answer. The transformed analyst is
causal: it reconstructs the Pauli answers from the effect estimates
and then follows the original next-query rule.

For Hilbert-space dimension $2^n$, the parameter in those protocols obeys
\begin{equation}
 \kappa=\Theta(\log(2^n)/\varepsilon_{\mathrm{eff}}^2)
 =\Theta(n/\varepsilon^2).
 \label{sm:eq:kappa}
\end{equation}
Theorems 1.2 and 1.3 of Ref.~\cite{COPW26}, at constant joint failure
probability, then give the respective copy bounds
\begin{equation}
 O\!\left(\frac{\sqrt n}{\varepsilon^3}
       \log\!\left(M+\frac n{\varepsilon^2}+1\right)\right),
 \qquad O\!\left(\frac{\sqrt M}{\varepsilon^2}\right),
 \label{sm:eq:onlinebranches}
\end{equation}
with joint success at least $9/10$. Combining them with the published
full-table protocol gives, for the range of Theorem~\ref{sm:thm:precision},
\begin{equation}
 \Ns(n,M,\varepsilon)\le C_1\min\left\{
 \frac{\sqrt M}{\varepsilon^2},\frac n{\varepsilon^4},
 \frac{\sqrt n}{\varepsilon^3}
 \log\!\left(M+\frac n{\varepsilon^2}+1\right)\right\}.
 \label{sm:eq:precisionupper}
\end{equation}
The choice of protocol uses only the public parameters and is made
before receiving the state. Thus it gives a minimum of copy bounds.
At $\varepsilon=\eps$, Eq.~\eqref{sm:eq:onlinebranches} becomes
\begin{equation}
 \Ns(n,M)\le C\min\{\sqrt M,\sqrt n\log(M+n+1)\}.
 \label{sm:eq:online}
\end{equation}
The fixed-accuracy matching for $M\le n$ follows already from the
published bound \eqref{sm:eq:published}. Chen et al.'s online shadow-tomography study gives the stronger
online branches in Eqs.~\eqref{sm:eq:precisionupper} and \eqref{sm:eq:online}.
The lower and upper bounds retain a gap in their dependence on
$\varepsilon$.

\section{Consequences across query budgets}\label{sm:sec:regimes}
Combine Theorem~\ref{sm:thm:fixed} and Eq.~\eqref{sm:eq:identify}, and combine
the three upper protocols from Sec.~\ref{sm:sec:upper}. With
$r=\sqrt{\min\{n,M\}}$ and $\ell=\min\{n,\log(M+1)\}$, there are
universal constants $c',C'>0$ such that
\begin{align}
 \Ns(n,M)&\ge c'\max\{r,\ell\},
 \label{sm:eq:combinedlower}\\
 \Ns(n,M)&\le C'\min\{\sqrt M,\sqrt n\log(M+n+1),n\}.
 \label{sm:eq:combinedupper}
\end{align}
The lower maximum constrains one universal mechanism, even though its
hard states differ. The upper minimum selects a protocol before the
state is received.

For $1\le M\le n$, Theorem~\ref{sm:thm:fixed} and the published bound
\eqref{sm:eq:published} give
$\Ns(n,M)=\Theta(\sqrt M)$. In particular, $M=n$ costs
$\Theta(\sqrt n)$ copies, while a fixed list of $n$ queries costs
$O(\log n)$ at the same fixed accuracy and confidence.
For $M=\lceil n^\alpha\rceil$ with fixed $0\le\alpha\le1$,
$n^\alpha\le M\le2n^\alpha$ and $M\le n$, so
$\Ns(n,M)=\Theta(n^{\alpha/2})$. For fixed $\alpha>1$,
Eq.~\eqref{sm:eq:combinedupper} and Theorem~\ref{sm:thm:fixed} give
$c\sqrt n\le\Ns(n,\lceil n^\alpha\rceil)
\le C_\alpha\sqrt n\log n$. Dividing logarithms by $\log n$
removes the constant terms and $\log\log n/\log n$, proving
\begin{equation}
 \lim_{n\to\infty}
 \frac{\log\Ns(n,\lceil n^\alpha\rceil)}{\log n}
 =\frac12\min\{\alpha,1\}\qquad(\alpha\ge0).
 \label{sm:eq:exponent}
\end{equation}
For $\alpha=0$, the numerator is bounded and the same limit is zero.

If $M\ge2^{\gamma n}$ with fixed $\gamma>0$, then
$\ell\ge\min\{1,\gamma\}n$. Equation~\eqref{sm:eq:identify} gives
$\Ns(n,M)\ge\min\{1,\gamma\}n/3$, and the full table gives an
$O(n)$ upper bound. Thus $\Ns(n,M)=\Theta_\gamma(n)$.
Other budgets can retain a gap. For example,
$M=\lceil n\log n\rceil$ yields
$\Omega(\sqrt n)\le\Ns(n,M)\le O(\sqrt{n\log n})$.
Table~\ref{sm:tab:regimes} distinguishes constant-factor matching from
the polynomial-exponent statement.

\begin{table}[ht]
\caption{Fixed-accuracy copy bounds with simultaneous success at least
$2/3$. The exponents $\alpha>1$ and $\gamma>0$ are fixed.}
\label{sm:tab:regimes}
\begin{ruledtabular}
\begin{tabular}{lll}
Query budget & Lower bound & Upper bound\\
$1\le M\le n$ & $\Omega(\sqrt M)$ & $O(\sqrt M)$\\
$M=\lceil n^\alpha\rceil$, $\alpha>1$
 & $\Omega(\sqrt n)$ & $O_\alpha(\sqrt n\log n)$\\
$M\ge2^{\gamma n}$ & $\Omega_\gamma(n)$ & $O(n)$
\end{tabular}
\end{ruledtabular}
\end{table}

\section{Comparison with classical precision bounds}\label{sm:sec:classical}
For a bounded classical query $g_j$, let $p_{\mathrm{pop}}$ be the population
distribution and $g_j(\mathsf S)$ its empirical mean on a sample
$\mathsf S$. The relevant errors are
\begin{equation}
 \begin{aligned}
 \text{population accuracy:}&\quad |a_j-\mathbb E_{p_{\mathrm{pop}}}g_j|,\\
 \text{sample accuracy:}&\quad |a_j-g_j(\mathsf S)|,\\
 \text{post hoc generalization:}&\quad
 |g_j(\mathsf S)-\mathbb E_{p_{\mathrm{pop}}}g_j|.
 \end{aligned}
 \label{sm:eq:classicalcriteria}
\end{equation}
These impose different conditions. Here post hoc generalization is used
in the interactive-query sense of Definition 2.1 in Ref.~\cite{NSSSU18},
with a simultaneous guarantee over the queried functions for every
population distribution and adaptive analyst. On any event on which
all answers are both sample-accurate and population-accurate to
tolerance $\varepsilon$, the triangle inequality bounds every
corresponding discrepancy between sample and population means by
$2\varepsilon$. The resulting guarantee holds with at least the
probability of that joint event. Population accuracy alone does not
require sample accuracy.
The task in Sec.~\ref{sm:sec:task} requires population accuracy.

Published Theorem 1.2 of Nissim \emph{et al.}~\cite{NSSSU18} gives
$\Omega(\sqrt M/\varepsilon^2)$ samples for unrestricted statistical
queries under both population accuracy and post hoc generalization.
The population-only Theorem 13 in Lyu and Talwar's study of lower bounds
for adaptive data analysis~\cite{LT25} gives
$\Omega(\sqrt M/[\varepsilon\log(1/\varepsilon)])$ in its stated
parameter and domain regime, with universe size
$2^{O(\varepsilon^2M)}$. Removing the residual precision gap in this
dimension-independent population-only comparison remains open, as
discussed in Ref.~\cite{LT25} and Sec.~1.1 of Ref.~\cite{COPW26}.

\subsection{The finite-universe lower bound and its regime}\label{sm:sec:ltfinite}
Theorem 12 in Lyu and Talwar's study of lower bounds for adaptive data
analysis~\cite{LT25} concerns general bounded
queries with population accuracy alone. Its displayed rate is
\begin{equation}
 \Omega\!\left(\frac{\sqrt{\log U}\,\log M}
 {\varepsilon^2\log(1/\varepsilon)}\right),\qquad U=|\mathcal X|,
 \label{sm:eq:lt12rate}
\end{equation}
with $O(\log U)$ adaptive rounds and a positive constant probability
of forcing an inaccurate answer. Theorem 12 leaves this constant
unspecified. To contradict the $2/3$ success requirement in
Sec.~\ref{sm:sec:task}, a failure probability strictly greater than
$1/3$ would be needed; that numerical guarantee is not part of the
cited statement.
The printed assumptions include $U,M\ge1$, $0<\varepsilon<1/10$,
$M\ll2^U$, $U\ll2^M$, and $1/\varepsilon\ll\min\{U,M\}$.
To apply the rate, one must also fit the construction parameters in
the chosen domain and query budget. Write $d_{\rm LT},m_{\rm LT}$
for its two integer dimensions and $h_{\rm LT}=\Theta(\varepsilon^{-2})$
for its auxiliary dimension. Sections 3.1--3.2 and 4.4--4.5 of that
study use a base universe and total query count
\begin{equation}
 U_0=2^{d_{\rm LT}}m_{\rm LT}h_{\rm LT},\qquad
 Q_0=2^{m_{\rm LT}}h_{\rm LT}d_{\rm LT},\qquad
 W=\left(\frac{\log U_0\log Q_0}{\varepsilon^3}\right)^{O(1)}.
 \label{sm:eq:ltparameters}
\end{equation}
The obfuscated universe has $U_0W$ points. The conditions
$U_0W\le U$ and $Q_0\le M$ suffice to embed the constructed attack;
unused labels and queries can be padded. Recovering the rate in
Eq.~\eqref{sm:eq:lt12rate} additionally requires
$d_{\rm LT}=\Theta(\log U)$ and $m_{\rm LT}=\Theta(\log M)$.
The total count $Q_0$ already includes the $d_{\rm LT}$ stages, each
with $2^{m_{\rm LT}}h_{\rm LT}$ queries. Proposition 5 and the
population-only reduction in Sec.~4.4.2 give the rate
$\Omega(\sqrt{d_{\rm LT}}\,m_{\rm LT}/[\varepsilon^2\log(1/\varepsilon)])$.
These constraints matter: for $U=2^n$ and $M=\lceil n\log n\rceil$
at fixed accuracy, $d_{\rm LT}=\Theta(n)$ and $Q_0\le M$ force
$m_{\rm LT}=O(\log\log n)$, so substitution of a full $\log M$
factor is not justified there.

We now instantiate their construction for each fixed $\alpha>1$ and
fixed $0<\varepsilon<1/10$, with $U=2^n$ and
$M=\lceil n^\alpha\rceil$.
For sufficiently large $n$, choose
\begin{equation}
 d_{\rm LT}=\lfloor n/2\rfloor,\qquad
 m_{\rm LT}=\left\lfloor\frac{\alpha-1}{2}\log_2 n\right\rfloor,
 \qquad h_{\rm LT}=\Theta(\varepsilon^{-2}).
 \label{sm:eq:ltpolynomial}
\end{equation}
Round $h_{\rm LT}$ to an available Hadamard dimension, changing only
constant factors. Then
$Q_0=O_\varepsilon(n^{(\alpha+1)/2})\le\lceil n^\alpha\rceil$.
Moreover, $\log U_0=\Theta(n)$ and
$\log Q_0=\Theta_\alpha(\log n)$, so
\begin{equation}
 U_0W=2^{\lfloor n/2\rfloor}\operatorname{poly}_{\varepsilon,\alpha}(n)
 \le2^n.
 \label{sm:eq:ltuniverse}
\end{equation}
With the same constant failure guarantee, this gives
$\Omega_{\varepsilon,\alpha}(\sqrt n\log n)$ samples for general
bounded queries. This parameter specialization is our deduction from
the cited construction. Its constants need not be uniform as
$\alpha$ approaches one, and it inherits the source's failure
probability without converting it to the $2/3$-success task.

\subsection{What a transfer to parity queries would require}
The complete Lyu--Talwar attack composes each base query with the
inverse of a random, prefix-dependent input obfuscation and sets it
to $+1$ on points selected using previous answers. Their reduction
does not show that these operations preserve single parities on
$O(n)$ fixed bits. The static base queries alone admit elementary
parity encodings, but those encodings do not provide the information
hiding used in the adaptive reduction. A short binary label for a
universe point likewise does not make every bounded function a parity.

A transfer would require a single initial encoding on $O(n)$ bits,
single-parity queries at every stage, a causal query rule, and
preservation of the mechanism's information and population means.
If it also established failure probability strictly greater than $1/3$, it
would transfer the $\Omega_\alpha(\sqrt n\log n)$ fixed-accuracy rate
for every fixed $\alpha>1$ and close the logarithmic interval in
Sec.~\ref{sm:sec:regimes}. This is a conditional open problem, not a
parity lower bound established here.

For parities,
Ref.~\cite{HLHP26} already gives a finite-dimensional lower bound via
diagonal states and $Z$ queries. The additive encoding improves that
bound's dependence on the number of bits per sample, from the
$\min\{\sqrt M,n^{1/3}\}$ baseline to $\sqrt{\min\{n,M\}}$ at fixed
accuracy. The restricted parity representation and the
population-only precision dependence are distinct parts of this
comparison.

\clearpage
\endgroup
\end{document}